\documentclass[obeyspaces,preliminary,copyright,creativecommons]{eptcs}
\providecommand{\event}{ICE 2023} 

\PassOptionsToPackage{obeyspaces}{url}

\usepackage{iftex}

\ifpdf
  \usepackage{underscore}         
  \usepackage[T1]{fontenc}        
\else
  \usepackage{breakurl}           
\fi

\usepackage{localdef}

\usepackage{amsmath}
\usepackage{amsfonts}
\usepackage{amssymb}
\usepackage{latexsym}

\usepackage{graphicx}

\usepackage{listings}
\usepackage{tikz}
\usetikzlibrary{arrows,shapes,snakes,backgrounds}

\usepackage{wrapfig}
\title{On the Introduction of Guarded Lists in Bach:\\ Expressiveness, Correctness, and Efficiency Issues}
\author{Manel Barkallah 
\institute{Nadi Research Institute\\ Faculty of Computer Science\\ University of Namur \\ Namur, Belgium}
\email{manel.barkallah@unamur.be}
\and
Jean-Marie Jacquet 
\institute{Nadi Research Institute\\ Faculty of Computer Science\\ University of Namur \\ Namur, Belgium}
\email{\quad jean-marie.jacquet@unamur.be}
}

\def\titlerunning{On the Introduction of Guarded Lists in Bach: Expressiveness, Correctness, and Efficiency Issues}
\def\authorrunning{M. Barkallah \& J-M. Jacquet}

\hypersetup{
  bookmarksnumbered,
  pdftitle    = {\titlerunning},
  pdfauthor   = {\authorrunning},
  pdfsubject  = {EPTCS},               
}

\begin{document}
\maketitle

\begin{abstract}
Concurrency theory has received considerable attention, but mostly in
the scope of synchronous process algebras such as CCS, CSP, and
ACP. As another way of handling concurrency, data-based coordination
languages aim to provide a clear separation between interaction and
computation by synchronizing processes asynchronously by means of 
information being available or not on a shared space. Although these languages
enjoy interesting properties, verifying program correctness remains
challenging. Some works, such as Anemone, have introduced facilities,
including animations and model checking of temporal logic formulae, to
better grasp system modelling. However, model checking is known to
raise performance issues due to the state space explosion problem. In
this paper, we propose a guarded list construct as a solution to
address this problem. We establish that the guarded list construct
increases performance while strictly enriching the expressiveness of
data-based coordination languages. Furthermore, we introduce a notion
of refinement to introduce the guarded list construct in a
correctness-preserving manner.
\end{abstract}


\section{Introduction}

Concurrency theory has been the attention of a considerable effort
these last decades. However most of the effort has been devoted to
algebra based on synchronous communication, such as
CCS~\cite{Milner-CCS}, CSP~\cite{Hoare-CSP} and
ACP~\cite{Baeten-ACP}. Another path of research has been initiated by
Gelernter and Carriero, who advocated in \cite{GC92} that a clear
separation between the interactional and the computational aspects of
software components has to take place in order to build interactive
distributed systems. Their claim has been supported by the design of a
model, Linda \cite{CG89}, originally presented as a set of inter-agent
communication primitives which may be added to almost any programming
language. Besides process creation, this set includes primitives for
adding, deleting, and testing the presence/absence of data in a shared
dataspace. In doing so they proposed a new form of synchronization of
processes, occurring asynchronously, through the availability or
absence of pieces of information on a shared space.

A number of other models, now referred to as coordination models, have
been proposed afterwards. However, although many pieces of work have
been devoted to the proposal of new languages, semantics and
implementations, few articles have addressed the concerns of
practically constructing programs in coordination languages, in
particular in checking that what is described by programs actually
corresponds to what has to be modelled.

Based on previous results
\cite{BJ98,BJ03b,DJL13,DJLScp13,DJL14,DJL18,%
  JDB00,JL09,LJ04,LJ07,LJB04}, we have introduced in \cite{Scan} a
workbench \Scan\ to reason on programs written in \Bach, a Linda-like
dialect developed by the authors. It has been refined in
\cite{Anemone} to cope with relations, processes and multiple scenes.
The resulting workbench is named \Anemone. In both cases, one of our
goals was to allow the user to check properties by model checking
temporal logic formulae and by producing traces that can be replayed
as evidences of the establishment of the formulae. However, as
well-known in model checking, this goal raises performance issues
related to the state space explosion. In particular, letting
animation-related primitives interleave in many ways duplicates
research paths during model checking, with considerable performance
problems to check that formulae are established. To address this
problem, we introduce in this paper a guarded list construct and
establish that it yields an increase in performance while strictly
enriching the expressiveness of \Bach.

The rest of the paper is organized as follows. Section~\ref{anim-bach}
presents the reference Linda-like language \Bach\ employed by
\Scan\ and \Anemone. Section~\ref{guarded-list} introduces the guarded list
construction as well as the refinement relation. It is proved to
increase the expressiveness of the \Bach\ language in
Section~\ref{expressiveness} while the gain of efficiency in
model-checking is established in Section~\ref{performance}. Finally,
Section~\ref{relatedWork} compares our work with related work and
Section~\ref{conclusion} sums up the paper and sketches future work.

It is worth observing that, as duly compared in
Section~\ref{relatedWork}, introducing an atomic construct is not
new. However, our contribution is (i) to introduce a construct
tailored to coordination languages, (ii) to establish that it yields a
gain of performance in model checking and also an increase of
expressiveness, and finally (iii) to identify refinement-based
criteria so as to guide the programmer to introduce the guarded list
construct in a correctness-preserving manner.

\begin{figure}[t]
  \begin{center}
  \includegraphics[width=0.4\textwidth]{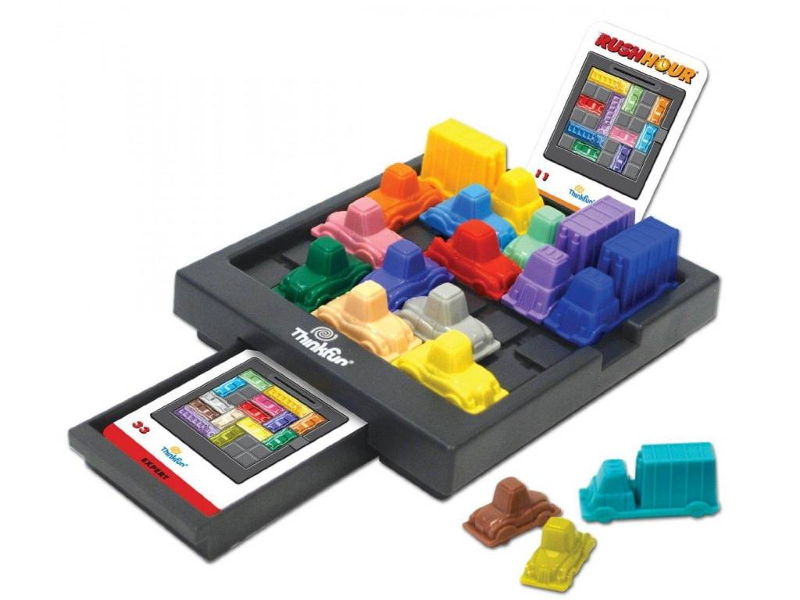}
  \hspace*{2cm}
  \begin{minipage}[t]{4cm}
\begin{tikzpicture}

\draw[step=0.5cm, gray] (0,0) grid (3,3);

\foreach \a in {1,2,...,6}{
\draw (-0.25,3.25-\a*0.50) node{$\a$};
\draw (\a*0.50 - 0.25,3.25) node{$\a$};}

\draw[fill=yellow] (0.1,2.9) rectangle (0.4,1.6); 
\draw[fill=white!30!purple] (1.6,2.9) rectangle (1.9,1.6);   
\draw[fill=white!35!blue] (1.6,1.4) rectangle (2.9,1.1); 
\draw[fill=black!35!green] (1.1,0.4) rectangle (2.4,0.1); 

\draw[fill=red] (0.6,1.9) rectangle (1.4,1.6); 
\draw[fill=white!25!green] (1.1,1.4) rectangle (1.4,0.6); 
\draw[fill=white!25!orange] (2.6,0.9) rectangle (2.9,0.1); 

\end{tikzpicture}

  \end{minipage}
  \end{center}
  \caption{Rush Hour Problem. On the left part, the game as
    illustrated at https://www.michaelfogleman.com/rush. On the right
    part, the game modeled as a grid of $6 \times 6$, with cars and
    trucks depicted as rectangles of different colors.
    \label{rush-hour}}
\end{figure}
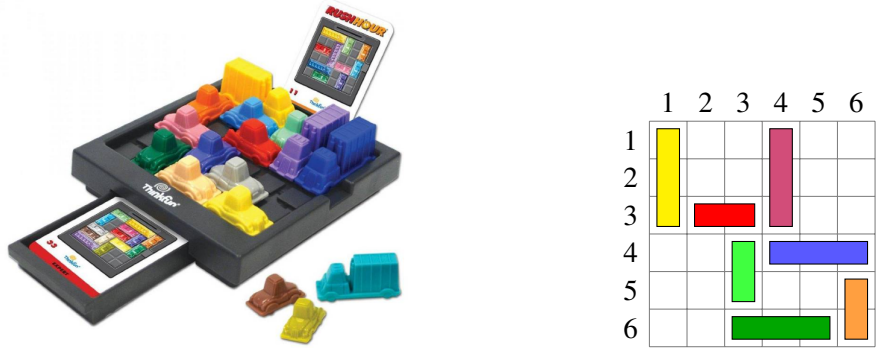

To make the article more concrete, we shall use the running example
of~\cite{Anemone}, namely a solution to the rush hour puzzle. This
game, illustrated in Figure~\ref{rush-hour}, consists in moving cars
and trucks on a $6 \times 6$ grid, according to their direction, such
that the red car can exit. It can be formulated as a coordination
problem by considering cars and trucks as autonomous agents which have
to coordinate on the basis of free places.

\section{The \AnimBach\ language}
\label{anim-bach}

\subsection{Definition of data}

Following Linda, the \Bach\ language \cite{DJL18,JL07} uses four
primitives for manipulating pieces of information: $tell$ to put a
piece of information on a shared space, $ask$ to check its presence,
$nask$ to check its absence and $get$ to check its presence and remove
one occurrence. In its simplest version, named \BachT, pieces of
information consist of atomic tokens and the shared space, called the
store, amounts to a multiset of tokens. Although in principle such a
framework is sufficient to code many applications, it is however too
elementary in practice to code them easily. To that end, we introduce
more structured pieces of information which may employ sets defined as
in

\begin{lstlisting}
eset RCInt = { 1, 2, 3, 4, 5, 6}.
\end{lstlisting}

\noindent
in which the set $RCInt$ is defined as the set containing the elements
$1$ to $6$. In addition to sets, maps can be defined between them as functions
that take zero or more arguments. In practice, mapping
equations are used as rewriting rules, from left to right in the aim of
progressively reducing a complex map expression into a set element.

As an example of a map, assuming a grid of 6 by 6 featuring the rush
hour problem as in \cite{Anemone} and assuming that trucks in this game
take three cells and are identified by the upper and left-most cell
they occupy, the operation \texttt{down\_truck} determines the cell to
be taken by a truck moving down:

\begin{lstlisting}
map down_truck : RCInt -> RCInt.
eqn down_truck(1) = 4. down_truck(2) = 5. down_truck(3) = 6. 
\end{lstlisting}

\noindent
Note from this example that mappings may be partially defined, with the
responsibility put on the programmer to use them only when defined.

Structured pieces of information to be placed on the store
consist of flat tokens as well as
expressions of the form $f(a_1, \cdots, a_n)$ where $f$ is a functor
and $a_1$, \ldots, $a_n$ are set elements or structured pieces of information. 
As an example, in the rush
hour example, it is convenient to represent the free places of the
game as pieces of information of the form \texttt{free(i,j)} with $i$
a row and $j$ a column.

The set of structured pieces of information is subsequently denoted by
${\cal I}$. For short, si-term is used later to denote a structured
piece of information. Mapping definitions induce a rewriting relation
that we shall subsequently denote by $\leadsto$, that rewrites
si-terms to final si-terms, namely si-terms that cannot be reduced further.

\subsection{Primitives}

The primitives consist of the \texttt{tell}, \texttt{ask},
\texttt{nask} and \texttt{get} primitives already introduced, which
take as arguments elements of ${\cal I}$. A series of graphical
primitives are added to them. They aim at animating the
executions. They include \texttt{draw}, \texttt{move\_to},
\texttt{place\_at}, \texttt{hide}, \texttt{show} primitives, to cite
only a few. The key point for this paper is that they always succeed
and do not interfere with the shared space. For the rest of the paper,
we shall assume a set $GPrim$ of graphical primitives and will take
primitives from it.  The coordinated \Bach\ language enriched by
graphical primitives is subsequently referred to as \AnimBach.

\begin{figure}[t]
\begin{center}
\begin{minipage}{15.5cm} 
%
\[ \begin{array}{ccc}
    \begin{array}{r@{\hspace*{0.25cm}}c}
     {\bf(T)}
& 
    $$\dfrac{t \leadsto u}{\langle \; tell(t) \; \vert \; \sigma \;\rangle \longrightarrow \langle \; E \;\vert \; \sigma \cup \lbrace u\rbrace \;\rangle}$$\\ 
\\ \\ 
 
     {\bf(A)}        
& 
    $$\dfrac{t \leadsto u}{\langle \; ask(t) \;\vert\; \sigma \cup \lbrace u \rbrace \;\rangle \longrightarrow \langle\; E \;\vert\; \sigma \cup \lbrace u\rbrace \;\rangle}$$\\ 
\\ \\ 
     {\bf(G)}        
& 
      $$\dfrac{t \leadsto u}{\langle\; get(t)\;\vert\; \sigma \cup \lbrace u \rbrace \;\rangle \longrightarrow \langle\; E \;\vert\; \sigma \;\rangle}$$\\
    \end{array} & \hspace*{1cm} & 
    \begin{array}{r@{\hspace*{0.25cm}}c}
      {\bf(N)}        
& 
      $$\dfrac{t \leadsto u, u \not\in \sigma}{\langle\; nask(t)\;\vert\; \sigma \;\rangle \longrightarrow \langle\; E \;\vert\; \sigma \;\rangle}$$\\
\\ \\
     {\bf(Gr)}
& 
    $$\dfrac{p \in GPrim}{\langle \; p \; \vert \; \sigma \;\rangle \longrightarrow \langle \; E \;\vert \; \sigma  \;\rangle}$$\\ 
    \end{array}
\end{array}
\] 
\end{minipage}
\end{center}

\caption{Transition rules for the primitives}
\label{jmj-fig-primitives}
\end{figure}

The execution of primitives is formalized by the transition steps of
Figure~\ref{jmj-fig-primitives}.  Configurations are taken there as
pairs of instructions, for the moment reduced to simple primitives,
coupled to the contents of the shared space. 
Following the constraint-like setting of \Bach\ in which the Linda primitives have been rephrased,
the shared space is renamed as \textit{store} and is formally defined
as a multiset of si-terms. 
As a result, rule (T) states that the
execution of the $tell(t)$ primitive amounts to enriching the store by
an occurrence of $t$. The $E$ symbol is used in this rule as well as in
other rules to denote a terminated computation. Similarly, rules (A)
and (G) respectively state that the $ask(t)$ and $get(t)$ primitives
check whether $t$ is present on the store with the latter removing one
occurrence. Dually, as expressed in rule (N), the primitive $nask(t)$
tests whether $t$ is absent from the store. Finally, rule (Gr)
expresses that any graphical primitive succeeds without modifying the
store.

\subsection{Agents}

Primitives can be composed to form more complex agents by using
traditional composition operators from concurrency theory: sequential
composition, parallel composition and non-deterministic choice. 
Another mechanism is added in \AnimBach: conditional statements of the form $c \rightarrow
s_1 \diamond s_2$, which computes $s_1$ if $c$ evaluates to true or
$s_2$ otherwise. As a shorthand, $c \rightarrow s_1$ is used to compute $s_1$ when $c$ evaluates to true.
Conditions of type $c$ are obtained from elementary
ones, thanks to the classical and, or and negation operators, denoted
respectively by $\&$, $|$ and $!$. Elementary conditions are obtained
by relating set elements or mappings on them by equalities (denoted
$=$) or inequalities (denoted $\!=$, $<$, $<=$, $>$, $>=$).

Procedures are defined similarly to mappings through the \texttt{proc}
keyword by associating an agent with a procedure name. As in classical
concurrency theory, it is assumed that the defining agents are guarded, in
the sense that any call to a procedure is preceded by the execution
of a primitive or can be rewritten in such a form.

As an example, the behavior of a vertical truck in the rush hour
puzzle can be modelled by the following code:

\begin{lstlisting}
proc VerticalTruck(r: RCInt, c: RCInt) =
  ( (r>1 & r<5) -> ( get(free(pred(r),c)); tell(free(succ(succ(r)),c);
                     VerticalTruck(pred(r),c) )
  +
  ( (r<4) -> ( get(free(down_truck(r),c)); tell(free(r,c));
               VerticalTruck(succ(r),c)) ).
\end{lstlisting}

\noindent
To understand it, remember that a truck is identified by the upper and
left-most cell it occupies. The parameters of the
\texttt{VerticalTruck} procedure are precisely the row number and the
line number of this cell.  Given that a vertical truck can move one
cell up or one cell down, the procedure offers two alternatives
through the "+" operator. The first one corresponds to a truck moving
one cell up. To make this move realistic, the row $r$ occupied by the
truck should be strictly greater than one. Otherwise, the truck is
already on the first row (like the yellow truck of
Figure~\ref{rush-hour}) and cannot move up. Moreover, as we shall see
in a few seconds, the row $r$ should also be strictly smaller than
5. Assuming the two conditions hold ($r>1 \& r<5$) moving a truck one
cell up proceeds in three steps. First we need to make sure that the
cell up is free. This is obtained by getting the si-term
$free(pred(r),c))$ by means of the execution of the
\texttt{get(free(pred(r),c))}
primitive. Note that $pred(r)$ is actualy coded by a map as being
$r-1$. Second the cell liberated by moving the truck one cell up is to
be declared free. This is obtained by telling the corresponding $free$
si-term on the store, namely by executing
\texttt{tell(free(succ(succ(r)),c)}.  Note that $succ(r)$ is coded as $r+1$
by a map, which is why $r$ needs to be smaller than 5. Third the truck
procedure has to be called recursively with $pred(r)$ and $c$ as new
coordinates for the upper and left-most cell it occupies.

The behavior of the alternative movement in which the truck goes down
by one cell is similar. As $r$ is assumed to be in set $RCInt = \{ 1,
\cdots, 6 \}$ and we do not perform a $pred$ operation there is no
need to check that $r$ is greater or equal to 1. However to get the
cell down we need to check that $r$ is strictly less than 4.

\begin{figure}[t]
\begin{minipage}{15.5cm}
\[ \begin{array}{ccc}
    \begin{array}{r@{\hspace*{0.25cm}}c}
     {\bf(S) }
   & 
     \infrulemath{  \transm{ \conf{ A }{ \sigma } 
                      }{ \conf{ A' }{ \sigma' } 
                       } 
            }{ 
                \transm{ \conf{ \seqcc{A}{B} }{ \sigma } 
                      }{ \conf{ \seqcc{A'}{B} }{ \sigma' } 
                       } 
              } 
   \\ \\ 
     {\bf(P) }
   & 
     \infrulemath{  \transm{ \conf{ A }{ \sigma } 
                      }{ \conf{ A' }{ \sigma' } 
                       } 
            }{ 
                 \begin{array}{c} 
                     \transm{ \conf{ \paracc{A}{B} }{ \sigma } 
                           }{ \conf{ \paracc{A'}{B} }{ \sigma' } 
                            } 
                  \\ 
                     \transm{ \conf{ \paracc{B}{A} }{ \sigma } 
                           }{ \conf{ \paracc{B}{A'} }{ \sigma' } 
                            } 
                  \end{array} 
              } 
    \\ \\ 
     {\bf(C) }
   & 
     \infrulemath{  \transm{ \conf{ A }{ \sigma } 
                       }{ \conf{ A' }{ \sigma' } 
                        } 
            }{ 
                 \begin{array}{c} 
                     \transm{ \conf{ \choicec{A}{B} }{ \sigma } 
                           }{ \conf{ A' }{ \sigma' } 
                            } 
                  \\ 
                     \transm{ \conf{ \choicec{B}{A} }{ \sigma } 
                           }{ \conf{ A' }{ \sigma' } 
                            } 
                  \end{array} 
              } \\ \\
\end{array} & \hspace*{1cm} &     \begin{array}{r@{\hspace*{0.25cm}}c}
     {\bf(Co) }
   & 
     \infrulemath{  \models C, \, \transm{ \conf{ A }{ \sigma } 
                       }{ \conf{ A' }{ \sigma' } 
                        } 
            }{ 
                 \begin{array}{c} 
                     \transm{ \conf{ C \rightarrow A \diamond B }{ \sigma } 
                           }{ \conf{ A' }{ \sigma' } 
                            } 
                  \\ 
                     \transm{ \conf{ !C \rightarrow B \diamond A }{ \sigma } 
                           }{ \conf{ A' }{ \sigma' } 
                            } 
                  \end{array} 
              } 
    \\ \\ 
     {\bf(Pc) }
   & 
     \infrulemath{  P(\overline{x}) = A, \transm{ \conf{ A[\overline{x}/\overline{u} }{ \sigma } 
                       }{ \conf{ A' }{ \sigma' } 
                        } 
            }{ 
                     \transm{ \conf{ P(\overline{u}) }{ \sigma } 
                           }{ \conf{ A' }{ \sigma' } 
                            } 
            } \\ \\
       \end{array}
\end{array} \]
\end{minipage}

\caption{Transition rules for the operators
\label{fig-operators}}
\end{figure}

The operational semantics of complex agents is
defined through the transition rules of
Figure~\ref{fig-operators}.
They are quite classical. Rules
(S), (P) and (C) provide the usual semantics for sequential, parallel
and choice compositions.  As expected, rule (Co) specifies that the
conditional instruction \( C \rightarrow A \diamond B \) behaves as
$A$ if condition $C$ can be evaluated to true and as $B$
otherwise. Note that the notation $\models C$ is used to denote the
fact that $C$ evaluates to true.  Finally, rule (Pc) makes procedure
call $P(\overline{u})$ behave as the agent $A$ defining procedure $P$
with the formal arguments $\overline{x}$ replaced by the actual ones
$\overline{u}$.

In these rules, it is worth noting that we assume agents of the form
($E ; A$), ($E \parac A$) and ($A \parac E$) to be rewritten as $A$.

\subsection{A fragment of temporal logic}

Linear temporal logic is widely used to reason on dynamic
systems. The \Scan\ and \Anemone\ workbenches use a fragment of PLTL
\cite{Eme90}.

As usual, the logic employed relies on propositional state
formulae. In the coordination context, these formulae are to be
verified on the current contents of the store. Consequently, given a
structured piece of information $t$, the notation $\#t$ is introduced to denote
the number of occurrences of $t$ on the store and  basic
propositional formulae are defined as equalities or inequalities combining algebraic
expressions involving integers and number of occurrences of structured
pieces of information. An example of such a basic formulae is
$\#free(1,1) = 1$ which states that the cell of coordinates $(1,1)$ is
free.

Propositional state formulae are built from these basic formulae by
using the classical propositional connectors. On the point of notations, given
a store $\sigma$ and a propositional state formulae $PF$, we shall write
$\sigma \models PF$ to indicate that $PF$ is established on store $\sigma$.

The fragment of temporal logic used in \Scan\ and \Anemone\ is
then defined from these propositional state formulae by the following grammar :
\[  TF ::= PF \, | \, \mathit{Next}\,\, TF \, | \, PF\,  \mathit{Until}\,\, TF
\]
where $PF$ is a propositional formula. A classical use, on which we
shall focus in this paper, is to determine whether a propositional
state formulae can be reached at some state.  As an example, coming
back to the rush hour problem, if the red car indicates that it leaves
the grid by placing $out$ on the store, a solution to the rush problem
is obtained by verifying the formula
\[  \mathit{true}\;  \mathit{ Until } \, (\#out=1)
\]
which we shall subsequently abbreviate as
\( \mathit{Reach} (\#out=1)
\).

The algorithm used in \Scan\ and \Anemone\ to establish reach
properties basically consists of a breadth-first search on the state
space engendered by an agent starting from the empty store. During
this search, for each newly created state, a test is made to check
whether the considered reach property holds.

Such an elementary algorithm works well for simple problems. However it
becomes difficult to use when more complex problems are tackled. One
of the reasons comes from the fact that states are duplicated many times by
interleaving. Consider for instance the code for the
\texttt{VerticalTruck} procedure introduced above. With primitives to
animate its execution and colors introduced for visualization purposes, its more
complete code is as follows:

\label{page-listing-vertical-truck}
\begin{lstlisting}
proc VerticalTruck(r: RCInt, c: RCInt, p: Colors) =
        ( (r>1 & r<5) ->  ( get(free(pred(r),c));
	                    moveTruck(pred(r),c,p);
	                    tell(free(succ(succ(r)),c));
                            VerticalTruck(pred(r),c,p) ))
        +
        ( (r<4) -> ( get(free(down_truck(r),c));
	             moveTruck(succ(r),c,p);
	             tell(free(r,c));
                     VerticalTruck(succ(r),c,p) )).
\end{lstlisting}

\noindent
\begin{sloppypar}
Consider now two vertical trucks in parallel and for illustration the
first three statements: \texttt{get(free(pred(r),c))}, \texttt{moveTruck(pred(r),c,p)}
\texttt{tell(free(succ(succ(r)),c))}. 
Interleaving them in the two parallel instances
of \texttt{VerticalTruck} is of no interest for checking whether $out$
has been produced since what really matters is the state resulting
after the three steps. Hence, provided the first get primitive
succeeds, the two other primitives may be executed in a row. This
observation leads us to introduce so-called guarded lists of
primitives.
\end{sloppypar}

\section{A guarded list construct}
\label{guarded-list}

A \textit{guarded list} of primitives is a construct of the form $ \;
[ p \rightarrow p_1, \cdots, p_n ] \; $ where $p$, $p_1$, \ldots,
$p_n$ are primitives, with the list $p_1$, \ldots, $p_n$ being
possibly empty. In that latter case, we shall write $[ p ]$ for
simplicity of the notations.

Basically, a guarded list of primitives is a list of primitives
containing at least one primitive. The reason for writing guarded
lists with an arrow and for calling it guarded comes from the fact
that, provided the first primitive can be successfully executed, all
the others are executed immediately after without rollback in case of
failure. It is of course the responsibility of the programmers to
guarantee that in case the first primitive can be successfully
evaluated the remaining primitives can also be successfully
executed. Note that this is obviously the case for tell primitives and
the graphical primitives which always succeed regardless of the
current content of the store. Note also that we shall subsequently
identify criteria to introduce guarded lists while preserving
correctness.

It is worth observing that guarded lists are atomic constructs which
makes them different from conditional statements. In two words, the
execution of $\; [ p \rightarrow p_1, \cdots, p_n ] \; $ is as
follows.  First the store is locked and the execution of $p$ is
tested. If it fails then no modification is performed on the store and
the store is released. Otherwise not only $p$ is executed but also
after $p_1$, \ldots, $p_n$ in a row. After that the store is
released. In contrast, the execution of the conditional statement
$c \rightarrow s_1 \diamond s_2$ amounts to check $c$, which does not
require to lock the store since conditions are built on comparing
si-terms and not their presence or absence on the store. If $c$ is
evaluated to true then $s_1$ is executed, which means that one step of
$s_1$ is done if this is possible. If $c$ is evaluated to false then
one step of $s_2$ is attempted.

\begin{figure}[t]
\begin{center}
\begin{minipage}{15.5cm}
%
\[ \begin{array}{r@{\hspace*{0.25cm}}c}
     {\bf(Le)}
& 
    $$\langle \; [] \; \vert \; \sigma \;\rangle \longrightarrow \langle \; E \;\vert \; \sigma \;\rangle$$\\ 
\\ \\ 
 
     {\bf(Ln)}        
& 
     $$\dfrac{ \langle \; p \;\vert\; \sigma \;\rangle \longrightarrow \langle\; E \;\vert\; \tau \;\rangle, \;
               \langle \; L \;\vert\; \tau \;\rangle \longrightarrow^{*} \langle\; E \;\vert\; \phi \;\rangle
       }{ \langle \; [ p | L ] \;\vert\; \sigma \;\rangle \longrightarrow \langle\; E \;\vert\; \phi \;\rangle}$$\\ 
\\ \\
     {\bf(GL)}
& 
     $$\dfrac{ \langle \; p \;\vert\; \sigma \;\rangle \longrightarrow \langle\; E \;\vert\; \tau \;\rangle, \;
               \langle \; L \;\vert\; \tau \;\rangle \longrightarrow^{*} \langle\; E \;\vert\; \phi \;\rangle
       }{ \langle \; [ p \rightarrow L ] \;\vert\; \sigma \;\rangle \longrightarrow \langle\; E \;\vert\; \phi \;\rangle}$$\\ 
\end{array}
\] 
\end{minipage}
\end{center}

\caption{Transition rules for guarded lists}
\label{fig-guarded-list}
\end{figure}

The operational semantics of guarded lists is defined by rules (Le),
(Ln) and (GL) of Figure~\ref{fig-guarded-list}. The first two rules
define the semantics of lists of primitives, as being successively
executed. Rule (Le) concerns the empty list of primitives $[ \, ]$
while rule (Ln) inductively specifies that of a non-empty list $[p|L]$
with $p$ the first primitive and $L$ is the list of the other
primitives\footnote{%
  These list notations $[ \, ]$ and $[p|L]$ come from the
  logic programming way of handling lists.}.
Rule (GL) then states that the guarded list $[ p \rightarrow L ]$
can do a computation step from the store $\sigma$
to $\phi$ provided the primitive $p$ can do a step changing the store
$\sigma$ to $\tau$ and provided the list of primitives $L$ can
change $\tau$ to $\phi$.

Of course, introducing guarded lists as an atomic construct reduces
the interleaving possibilities between parallel processes. This is in
fact what we want to achieve to get speed ups in the model checking
phase. However from a programming point of view, one needs to
guarantee that computations are kept in some way. This is the purpose
of the introduction of the histories and of their contractions.

\begin{sloppypar}
\begin{definition}
\mbox{ }
  
\begin{enumerate}

\item Define the set of computational histories (or histories for
  short) $\Shist$ as the set \( \Sstore^\infseq \cup
  \Sstore^\fseq.\{\delta^+,\delta^-\} \) where $Sstore$ denotes the
  set of stores (namely of finite multisets of final si-terms), the
  $\fseq$ and $\infseq$ symbols are used to respectively denote finite
  and infinite repetitions and where $\delta^{+}$ and $\delta^{-}$ are
  used as ending marks respectively denoting successful and failing
  computations.

\item A history $h_c$ is a contraction of an history $h$ if it can be
  obtained from the latter by removing a finite number (possibly 0)
  elements of it, except the terminating marks $\delta^{+}$ and
  $\delta^{-}$. This is subsequently denoted by $h_c \preceq h$.

\item Given a contraction
\( h_c = \sigma_0 .  \cdots . \sigma_n . \delta \)
(resp. \( h_c = \sigma_0 .  \cdots . \sigma_n . \cdots \))
of an history $h$, there are thus sequences of stores, $\overline{\sigma_0}$, \ldots, $\overline{\sigma_n}$ such that
\( h = \overline{\sigma_0} . \sigma_0 .  \cdots . \overline{\sigma_n}. \sigma_n . \delta \)
(resp. \( h = \overline{\sigma_0} . \sigma_0 .  \cdots . \overline{\sigma_n}. \sigma_n . \cdots \)).
For any logic formula $F$, the history $h_c$ is said to be $F$-preserving iff, for any i and for any store 
$\tau$ of $\overline{\sigma_i}$, one has $\tau \models F$ iff $\sigma_i \models F$. This is subsequently denoted as
$h_c \ll_F h$.

\end{enumerate}
\end{definition}
\end{sloppypar}

Contractions and $F$-preserving contractions can be lifted in an obvious way to sets of histories.

\begin{definition}
A set $S_c$ of histories is a contraction (resp. a $F$-preserving
contraction) of a set $S$ of histories if any history of $S_c$ is the
contraction (resp. a $F$-preserving contraction) of a history of
$S$. By lifting notations on histories, this is subsequently
denoted by $S_c \preceq S$ (resp. $S_c \ll_F S$).
\end{definition}

We can now define the history-based operational semantics as the one delivering
all the computational histories. To make it general, we shall define it on
any contents of the initial store.

\begin{definition}
Define the language $\gdLg$ as the \AnimBach\ language with the guard list construct.
\end{definition}

\begin{definition}
Define the operational semantics 
$\OpSemh:\gdLg\rightarrow \pnc\Shist$ 
as the following function.  For any agent $A$ and any store $\tau$
\begin{eqnarray*}
    \lefteqn{\OpSemh(A)(\tau) = }  \\ & &
            \bset 
              \sigma_0 .  \cdots . \sigma_n. \delta^{+}:
                        \tconf{A}{\sigma_0} 
                            \longrightarrow \cdots \longrightarrow
                            \tconf{E}{\sigma_n},
                         \sigma_0 = \tau, n \geq 0
            \eset
        \\ & &
        \cup 
            \bset 
              \sigma_0 .  \cdots . \sigma_n . \delta^{-}:
                        \tconf{A}{\sigma_0}
                             \longrightarrow \cdots \longrightarrow
                             \tconf{A_n}{\sigma_n} 
                             \not\longrightarrow,
                         \sigma_0 = \tau, A_n \not= E, n \geq 0
            \eset
        \\ & &
        \cup
            \bset 
              \sigma_0 .  \cdots . \sigma_n . \cdots :
                        \tconf{A}{\sigma_0} 
                            \longrightarrow \cdots \longrightarrow
                            \tconf{A_n}{\sigma_n} \longrightarrow \cdots,
                        \sigma_0 = \tau, \forall n \geq 0 : A_n \neq E
            \eset
\end{eqnarray*}

\end{definition}

We are now in a position to define the refinement of agents.

\begin{definition}
Agent $A$ is said to refine agent $B$ iff 
\( \OpSemh(A)(\tau) \preceq \OpSemh(B)(\tau)
\),
for any store $\tau$.
\end{definition}

The following proposition is a direct consequence of the above
definitions.  Its interest is to establish contractions and
$F$-preserving properties from a syntactic characterization.

\begin{proposition}
\label{prop-preserving-contraction}
\mbox{ }
  
\begin{enumerate}

\item If $p_1, \cdots, p_n$ are tell primitives or
  graphical primitives then for any primitive $p$, the guarded list
  $GL = [ p \rightarrow p_1, \cdots, p_n ]$ refines the sequential
  composition $SC = p; p_1; \cdots; p_n$. As a result, any reachable
  property proved on the stores generated by the execution of $GL$ from a given store $\tau$
  is also established on the stores generated by the execution of $SC$ from $\tau$.

\item Assuming additionally that the arguments of the tell primitives
  of $p_1, \cdots, p_n$ are distinct from the si-terms appearing in
  the reachable formulae $F$, then $GL$ is also a $F$-preserving
  contraction of $SC$. It results that $F$ is established on the stores resulting from the execution of $SC$ from any store $\tau$ iff
  it is established on the stores resulting from the execution of $GL$ from $\tau$.

\end{enumerate}
\end{proposition}

For the study of expressiveness, it will be useful to turn to a
simpler semantics focusing on the resulting stores of finite
computations. Such a semantics is defined as follows.

\begin{definition}
\label{def-opsemf}
Define the {\em operational semantics\/} 
\(  \OpSemf: \gdLg \rightarrow \powerset{\Sstore \times \{\delta^{+}, \delta^{-}\} } \) 
as the following function: for any agent $A \in \gdLg$
\begin{eqnarray*} 
    \OpSemf(A) 
&   = 
&   \begin{array}[t]{l} 
       \bset  
           (\sigma,\delta^{+}) : 
           \conf{A}{\emptyset} \rightarrow^{*} \conf{E}{\sigma} 
       \eset 
    \\ 
    \cup 
    \\ 
       \mbox{} \bset  
           (\sigma,\delta^{-}) : 
           \conf{A}{\emptyset} \rightarrow^{*} \conf{B}{\sigma}  
                 \not\rightarrow,  
           B \not= E 
               \eset 
  \end{array} 
\end{eqnarray*} 
\end{definition}

It is immediate to verify that, for any agent $A$, the semantics 
$\OpSemf(A)$ is obtained by considering the final stores of the finite
histories of $\OpSemh(A)(\emptyset)$.

\section{Expressiveness}
\label{expressiveness}

Although it is interesting to bring efficiency during model checking, the
guarded list construct also brings an increase of expressiveness. This
is evidenced in this section by using the notion of modular embedding
introduced in \cite{BP94}. As pointed out there, from a computational
point of view, all ``reasonable'' sequential programming languages are
equivalent, as they express the same class of functions.  Still it is
common practice to speak about the ``power'' of a language on the
basis of the expressibility or non-expressibility of programming
constructs. In general, a sequential language $L$ is considered to be
more expressive than another sequential language $L'$ if the
constructs of $L'$ can be translated in $L$ without requiring a
``global reorganization of the program'' \cite{Fel90}, that is, in a
compositional way. Of course the translation must preserve the
meaning, at least in the weak sense of preserving termination.
 
When considering concurrent languages, the notion of termination must
be reconsidered as each possible computation represents a possible
different evolution of a system of interacting processes. Moreover
{\em deadlock} represents an additional case of termination. 
We shall consequently
rely on the operational semantics $\OpSemf$ of
Definition~\ref{def-opsemf}, focused on the final store of finite
computations together with the termination mark.

The basic definition of embedding, given by Shapiro \cite{Sha92} is 
the following. Consider two languages $L$ and $L'$. Moreover assume we are given the 
semantics mappings 
${\cal O}: L \rightarrow \mathit{Obs}$ and  
${\cal O'}: L' \rightarrow \mathit{Obs}'$, where $\mathit{Obs}$ and $\mathit{Obs}'$ are 
some suitable domains. 
Then $L$ can {\em embed} $L'$ if there exists a mapping $\cal C$ ({\em 
  coder}) from the statements of $L'$ to the statements of $L$, and 
a mapping $\cal D$ ({\em decoder}) from $\mathit{Obs}$ to $\mathit{Obs}'$, such 
that the diagram of Figure \ref{basic-embedding} commutes, namely such that  
for every statement $A \in L'$: 
\( {\cal D} ( {\cal O} ( {\cal C} (A) ) )  =  {\cal O}' ( A ) 
\).

 
\begin{figure}[t]
\begin{center} 
\begin{picture}(160,100)(-40,85) 
 
\put(  0,180){\makebox(0,0)[lb]{$L'$}} 
\put(  0,100){\makebox(0,0)[lb]{$L$}} 
\put( 80,180){\makebox(0,0)[lb]{$\mathit{Obs}'$}} 
\put( 80,100){\makebox(0,0)[lb]{$\mathit{Obs}$}} 
 
\put( 15,103){\vector(1,0){60}} 
\put( 15,183){\vector(1,0){60}} 
\put(  3,172){\vector(0,-1){60}} 
\put( 83,112){\vector(0,1){60}} 
 
\put( 40,187){\makebox(0,0)[lb]{${\cal O}'$}} 
\put( 40,90){\makebox(0,0)[lb]{${\cal O}$}} 
 
\put( -10,140){\makebox(0,0)[lb]{${\cal C}$}} 
\put( 87,140){\makebox(0,0)[lb]{${\cal D}$}} 
\end{picture} 
\end{center} 
\caption{Basic embedding.} 
\label{basic-embedding} 
\end{figure}
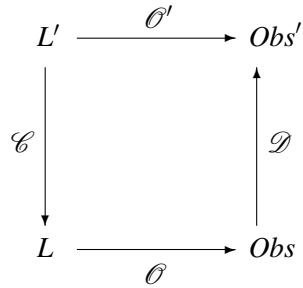 
 
 
The basic notion of embedding is too weak since, for instance, 
the above equation is satisfied by any pair of Turing-complete 
languages. 
De Boer and Palamidessi hence proposed in \cite{BP94} to add three 
constraints on the coder $\cal C$ and on the decoder $\cal D$ in order 
to obtain a notion of {\em modular} embedding usable for concurrent 
languages:

\begin{enumerate}

\item $\cal D$ should be defined in an element-wise way with respect
  to $\cal 
O$:  
\begin{labeleqn}{\mbox{($P_1$)}} 
\forall X \in \mathit{Obs}: \  
{\cal D}(X) = \{{\cal D}_{el}(x) \mid x \in X \} 
\end{labeleqn}%
for some appropriate mapping ${\cal D}_{el}$;

\item
the coder $\cal C$ should be defined in a compositional 
way with respect to the sequential, parallel and choice operators\footnote{%
Actually, this is not required for the sequential operator   
in \cite{BP94} since it does not occur in that work.}: 
\begin{labeleqn}{\mbox{($P_2$)}} 
\begin{array}{c} 
{\cal C}(A \seqc B) = {\cal C}(A) \seqc {\cal C}(B) \\ 
{\cal C}(A \parac B) = {\cal C}(A) \parac {\cal C}(B) \\ 
{\cal C}(A \choice B) = {\cal C}(A) \choice {\cal C}(B) \\ 
\end{array} 
\end{labeleqn}%

\item 
the embedding should preserve the behavior of the original 
processes with respect to deadlock, failure and success ({\em termination 
  invariance}): 
\begin{labeleqn}{\mbox{$(P_3$)}} 
\forall X \in \mathit{Obs}, \forall x \in X: \  
tm'({\cal D}_{el}(x)) = tm (x) 
\end{labeleqn}%
where $tm$ and $tm'$ extract the information on termination from the 
observables of $L$ and $L'$, respectively. 

\end{enumerate}

\noindent 
An embedding is then called {\em modular} if it satisfies properties  
$P_1$, $P_2$, and $P_3$. 

The existence of a modular embedding from $L'$ into $L$ 
is denoted as $L' \leq L$.
It is easy to see that $\leq$ is a pre-order relation.
Moreover if $L' \subseteq L$ then $L' \leq L$ that
is, any language embeds all its sublanguages. This property descends
immediately from the definition of embedding, by setting $\cal C$ and
$\cal D$ equal to the identity function.

Let us now compare the \AnimBach\ language with guarded
lists with the \AnimBach\ language without guarded lists. As
introduced before, the former is denoted by $\gdLg$. The latter will
be denoted by $\rLg$. 
Following~\cite{BJ03b}, we shall also test
three sublanguages composed (i) of the ask, tell primitives, (ii) of
the ask, tell, get primitives and (iii) of the ask, tell, get, nask
primitives. These sublanguages will be denoted by specifying the
primitives between parentheses, as in $\gdLg(\ask,\tell)$. Moreover,
to focus on the core features, we shall discard conditional statements
and procedures, which are essentially introduced for the ease of
coding applications.

By language inclusion, a first obvious result is that the
\AnimBach\ sublanguages with guarded lists embed their counterparts
without guarded lists.

\begin{proposition}
For any subset ${\cal X}$ of primitives, one has
\( \rLg({\cal X}) \leq \gdLg({\cal X}) \).
\end{proposition}

The converse relations do not hold. Intuitively, this is due to the
fact that, in contrast to $\rLg$, the languages $\gdLg$ have the
possibility of {\em atomically} testing the simultaneous presence of
two si-terms on the store. The formal proof requires of course a
deeper treatment. It turns out however that the techniques employed in
\cite{BJ03b} can be adapted to guarded lists. One of them, which
results from classical concurrency theory, is that any agent can be
reformulated in a so-called normal form.

\begin{definition} 
Agents (of $\gdLg$) in normal forms are agents of $\gdLg$ which obey the following grammar,
where  
$N$ is an agent in normal form, $p$ is a primitive (either graphical or store-related) or a guarded list of primitives
and $A$ denotes an arbitrary (non restricted) agent
\begin{eqnarray*} 
    N   & ::= &   p \ | \ \seqcc{p}{A} \ | \ \choicec{N}{N}.
\end{eqnarray*} 
\end{definition} 
 
\begin{proposition}
\label{equiv-normal-form}
For any agent $A$ of $\gdLg$, there is an agent $N$ of $\gdLg$ in normal form  which has the same derivation sequences as $A$.
\end{proposition} 

\begin{proof}
Indeed, it is possible to associate to any agent $A$ an agent $\tau(A)$ 
in normal form by using the following translation defined inductively on the structure of $A$:
\begin{eqnarray*}
      \tau(p) & = & p 
\\    \tau(X;Y) & = & \tau(X) ; Y
\\    \tau(X+Y) & = & \tau(X) + \tau(Y)
\\    \tau(X \parallel Y) & = & \tau(X) \leftmerge Y + \tau(Y) \leftmerge X
\\
\\    p \leftmerge Z & = & p; Z
\\    (p; A) \leftmerge Z & = & p; (A \parallel Z)
\\    (N_1 + N_2) \leftmerge Z & = & N_1 \leftmerge Z + N_2 \leftmerge Z
\end{eqnarray*}
It is easy to verify that, for any agent $A$, the agent $\tau(A)$ is
in normal form. Moreover, it is straightforward to verify that $A$ and
$\tau(A)$ share the same derivation sequences.
\end{proof}

We are now in a position to establish that $\gdLg(ask, tell)$ cannot be embedded in
$\rLg(ask, tell)$.

\begin{proposition}
\( \gdLg(ask, tell) \not\leq \rLg(ask, tell) \)
\end{proposition}

\begin{sloppypar}
\begin{proof}
Let us proceed by contradiction and assume the existence of a coder
$\coder$ and a decoder $\decoder$. The proof is composed of three main steps.

\medskip
\noindent
\textsc{Step 1:}\ on the coding of $tell(a)$ and $tell(b)$.
Let $a$, $b$ be two distinct si-terms. Since 
\( \OpSemf( \gtell{a} ) 
    = \bset ( \bset a \eset, \delta^{+} ) \eset
\),
any computation of $\coder(\gtell{a})$ starting in the empty store
succeeds by property $P_3$. Let 
\[         \conf{ \coder(\gtell{a}) }{ \emptystore }
           \longrightarrow
           \cdots
           \longrightarrow
           \conf{E}{ \bset a_1, \cdots, a_m \eset }
\]
be one computation of $\coder(\gtell{a})$. Similarly, any computation of
$\coder(\gtell{b})$ starting on the empty store succeeds. Let
\[         \conf{ \coder(\gtell{b}) }{ \emptystore }
           \longrightarrow
           \cdots
           \longrightarrow
           \conf{E}{ \bset b_1, \cdots, b_n \eset }
\]
be one computation of $\coder(\gtell{b})$. Note that, as we only
consider ask and tell primitives, this computations can be reproduced
on any store $\tau$. We thus have also that
\[         \conf{ \coder(\gtell{b}) }{ \tau }
           \longrightarrow
           \cdots
           \longrightarrow
           \conf{E}{ \tau \cup \bset b_1, \cdots, b_n \eset }
\]
In particular, as
\( \coder( \gtell{a}; \gtell{b} )
= \coder(\gtell{a}); \coder(\gtell{b})
\),
we have that
\[ \begin{array}[t]{@{}l}
       \conf{ \coder(\gtell{a}; \gtell{b} ) }{ \emptyset }
           \longrightarrow
           \cdots \\ \hspace*{2cm}
           \longrightarrow
           \conf{ \coder(\gtell{b} ) }{ \bset a_1, \cdots, a_m \eset } 
           \longrightarrow
           \cdots \\ \hspace*{2cm}
           \longrightarrow
           \conf{E}{ \bset a_1, \cdots, a_m, b_1, \cdots, b_n \eset }
   \end{array}
\]           

\medskip
\noindent
\textsc{Step 2:}\ coding of an auxiliary statement $AB$.
Consider now 
\( AB= \glist{ask(a)}{ask(b)} \).
Obviously, as it requires $a$ to be present, the execution of
$AB$ on the empty store cannot do any step and thus 
\( \OpSemf( AB ) = \bset ( \emptyset, \delta^{-} ) \eset
\).
Let us now turn to its coding $\coder(AB)$. By Proposition~\ref{equiv-normal-form}, it can be regarded in 
its normal form. As it is in $\rLg(tell,ask)$, its more general form is as follows
\[       tell(t_1) \seqc A_1 
     +   \cdots
     +   tell(t_p) \seqc A_p 
     +   ask(u_1) \seqc B_1
     +   \cdots
     +   ask(u_q) \seqc B_q
     +   gp_1 \seqc C_1
     +   \cdots
     +   gp_r \seqc C_r
\]
where $gp_1$, \ldots, $gp_r$ are graphical primitives.     
Let us first establish that there is no alternative guarded by a
$tell(t_i)$ operation. Indeed, if this was the case, then
\[     D =  \conf{\coder( AB )}{ \emptyset }
           \longrightarrow
           \conf{A_i}{ \bset t_i \eset }
\]
would be a valid computation prefix of 
$\coder(AB)$. As 
\( \OpSemf( AB ) = \bset ( \emptyset, \delta^{-} ) \eset
\), this prefix should deadlock afterwards.
However, as
\( \coder( AB + \gtell{a} ) = \coder(AB) + \coder(\gtell{a})
\),
the computation step $D$ is also a valid computation prefix of
\mbox{$\coder( AB + \gtell{a} )$}. Hence,
\mbox{$\coder( AB + \gtell{a} )$} admits a
failing computation which, by 
property $P_3$, contradicts the fact that
\( \OpSemf ( AB + \gtell{a} )
   =
   \bset ( \bset a \eset, \delta^{+} ) \eset 
\).
The proof of the absence of an alternative guarded by a graphical primitive $gp_i$ proceeds similarly.

Let us now establish that none of the $u_i$'s belong to
\( \bset a_1, \cdots, a_m \eset
   \cup
   \bset b_1, \cdots, b_n \eset
\).
Indeed, if $ u_j \in \bset a_1, \cdots, a_m \eset$
for some $j \in \bset 1, \cdots, q \eset$, then, as
\( \coder(\gtell{a}\seqc AB)
   = \coder(\gtell{a}) \seqc \coder(AB)
\),
the derivation
\[     D' =  \conf{ \coder(\gtell{a}\seqc AB) }{ \emptystore }
           \longrightarrow
           \cdots
     \begin{array}[t]{@{}l}
           \longrightarrow
           \conf{\coder(AB)}{ \bset a_1, \cdots, a_m \eset } \\
           \longrightarrow
           \conf{B_j}{ \bset a_1, \cdots, a_m \eset }
     \end{array}
\]
is a valid computation prefix of $\coder(\gtell{a}\seqc AB)$. However,
by applying rule (T),
\[ \conf{ \gtell{a}\seqc AB }{ \emptystore }
   \longrightarrow
   \conf{ AB }{ \{ a \} }
   \not\longrightarrow
\]
By Property $P_3$, it follows that $D'$ can
only be continued by failing suffixes. However,
thanks to the fact that
\( \coder(\gtell{a}\seqc (AB + \gask{a})) =
   \coder(\gtell{a}) \seqc (\coder(AB) + \coder(\gask{a}))
\) 
the prefix
$D'$ induces the following 
computation prefix $D''$ for
\( \coder(\gtell{a} \seqc ( AB + \gask{a} )) \)
\[     D'' =  
     \begin{array}[t]{@{}l}
           \conf{ \coder(\gtell{a}\seqc (AB + \gask{a})) }{ \emptystore }
           \longrightarrow
           \cdots \\ \hspace*{2cm}
           \longrightarrow
           \conf{\coder(AB)+\coder(\gask{a})}{ \bset a_1, \cdots, a_m \eset } \\
           \hspace*{2cm}
           \longrightarrow 
           \conf{B_j}{ \bset a_1, \cdots, a_m \eset } .
     \end{array}
\]
which can only be continued by failing suffixes whereas
\( \gtell{a} \seqc ( AB + \gask{a} )
\)
only admits a successful computation.

The proof proceeds similarly in the case
$ u_j \in \bset b_1, \cdots, b_n \eset$
for some $j \in \bset 1, \cdots, q \eset$
by then considering
\( \gtell{b} \seqc  AB \)
and
\( \gtell{b} \seqc ( AB + \gask{b} ) \).

\medskip
\noindent
\textsc{Step 3:} combining the first two steps to produce a contradiction. 
The $u_i$'s are thus forced not to belong to 
\( \bset a_1, \cdots, a_m \eset
   \cup
   \bset b_1, \cdots, b_n \eset
\). 
However, this induces a contradiction. To that end, let us first observe that
$\coder(AB)$ cannot do any step on the store $\bset a_1, \cdots, a_m, b_1, \cdots, b_n \eset$
since none of the $ask(u_i)$ primitives can do a step.
As a result,
\[ \conf{AB}{ \bset a_1, \cdots, a_m, b_1, \cdots, b_n \eset }
   \not\longrightarrow
\]
Now, by compositionality of the coder with respect to the sequential composition (property $P_2$),
\( \coder(\gtell{a} \seqc \gtell{b} \seqc AB)
   =
   \coder(\gtell{a}) \seqc \coder(\gtell{b}) \seqc \coder(AB)
\),
and consequently the following derivation is valid:
\[  \conf{ \coder(\gtell{a} \seqc \gtell{b} \seqc AB) }{ \emptystore }
           \longrightarrow
           \cdots
           \longrightarrow
           \conf{AB}{ \bset a_1, \cdots, a_m, b_1, \cdots, b_n \eset }
\]
and yields a failing computation for
\( \coder(\gtell{a} \seqc \gtell{b} \seqc AB)
\).
However, as easily checked, $\gtell{a} \seqc \gtell{b} \seqc AB$ has only 
one successful computation.
\end{proof}
\end{sloppypar}

Using similar arguments as in \cite{BJ03b}, it is possible to extend
the previous proof so as to establish the following results.
\begin{proposition}
\mbox{ } 
\begin{enumerate}
\item \( \gdLg(get, tell) \not\leq \rLg(get, tell) \)
\item \( \gdLg(ask, get, tell) \not\leq \rLg(ask, get, tell) \)    
\item \( \gdLg(ask, nask, get, tell) \not\leq \rLg(ask, nask, get, tell) \)    
\end{enumerate}
\end{proposition}

\section{Performance}
\label{performance}

Let us now illustrate the gain of efficiency during model-checking obtained by the guarded
list construct. To that end, we shall subsequently compare the
performance of the \Scan\ and \Anemone\ breath-first search model checker on various
examples of the rush hour puzzle coded, on the one hand, without
the guarded list construct, and, on the other hand, with the guarded list
construct.

As described in the previous sections, the rush hour puzzle can be
formulated as a coordination problem by considering cars and trucks as
autonomous agents which have to coordinate on the basis of free
places. The complete code is available 
at \cite{Rush-Hour-Program}. Besides sets, maps and widget definitions, it is basically
composed of generic procedures for coding horizontal cars and trucks
as well as vertical cars and trucks.  Specific cars and trucks are
then obtained by instantiating colors and places and are put in
parallel.

The code for the cars and trucks follows the pattern of the code
presented in page~\pageref{page-listing-vertical-truck}. Basically, under
some conditions, each car and truck amounts to (i) obtaining a free
place to move through the execution of a \texttt{get} primitive, (ii)
then to operating the movement graphically through the execution of a
\texttt{move} primitive and (iii) finally to freeing the place
previously occupied by means of the execution of a \texttt{tell} primitive.
As an example, the following code is a snippet refining the code of
page~\pageref{page-listing-vertical-truck}.

\begin{lstlisting}
get(free(pred(r),c));
move(truck_img(c),pred(r),c);
tell(free(succ(succ(r)),c))
\end{lstlisting}

The problem is solved when the $out$ si-term is put on the store,
which leads to checking that the property $\#out=1$ can be reached. 
As easily checked, the hypotheses of Proposition~\ref{prop-preserving-contraction} 
are verified so that we can replace the above code snippet by the following:

\begin{lstlisting}
[ get(free(pred(r),c)) ->
      move(truck_img(c),pred(r),c),
      tell(free(succ(succ(r)),c)) ]
\end{lstlisting}

\noindent
This code is indeed an $F$-preserving contraction for the formulae $F =
(\#out=1)$.

\begin{table}[t]
\begin{center}
\begin{tabular}{|c|c|p{11cm}|}
\hline
Case & Nb's cars/trucks & Game  \\
\hline
1  & 2 & VPurpleTruck(2,4), HRedCar(3,2) \\
\hline
2  & 3 & VPurpleTruck(2,1), HRedCar(3,2), HGreenCar(1,1) \\
\hline
3  & 4 & VPurpleTruck(2,1), HRedCar(3,2), HGreenCar(1,1), VOrangeCar(5,1) \\
\hline
4  & 5 & VPurpleTruck(2,1),  HRedCar(3,2), HGreenCar(1,1), VOrangeCar(5,1),  VBlueTruck(2,4) \\
\hline
5  & 6 & VPurpleTruck(2,1),  HRedCar(3,2), HGreenCar(1,1), VOrangeCar(5,1),  VBlueTruck(2,4),  HGreenTruck(6,3) \\
\hline
6  & 7 &  VPurpleTruck(2,1),  HRedCar(3,2), HGreenCar(1,1), VOrangeCar(5,1),  VBlueTruck(2,4),  HGreenTruck(6,3), VYellowTruck(1,6)  \\

\hline
\end{tabular}
\end{center}
\caption{Test cases \label{test-cases}}
\end{table}

\begin{table}[t]
\begin{center}
\begin{tabular}{|l|c|c|c|c|}
\hline
Case & Without GL & With GL & Gain  & Expected gain\\
\hline 
1  &   2630 ms (2s)  &   298 ms (0s) &  8.82 & 4  \\ 
\hline 
2 &  64341 ms (64s |1m) & 355 ms (0s)&  181 & 8 \\ 
\hline 
3  & 60339 ms (60s |1m)   &770 ms (1s) & 78 & 16 \\ 
\hline 
4  &   495578 ms (496s | 8m)  & 1032 ms (1s) & 480 & 32 \\ 
\hline 
5  & 3271343 ms (3271s | 55m)  & 4100 ms (4s)& 797 & 64 \\ 
\hline 
6  & $\geq$ 10h  &4862322 (1h35m)  &  $\geq$ 6 & 128 \\ 
\hline
\end{tabular}
\end{center}
\caption{Performance results
\label{guarded-list-performance}}
\end{table}

By performing this transformation, one gains per vehicle the
computation of two stores on four, which induces the hope of a gain of
performance of $2^n$ if $n$ is the number of vehicles in parallel. To
verify the actual gain of performance, we have model checked the two
codes (one with guarded list and the other without guarded list) on
the examples of Table~\ref{test-cases}. They are inspired by cards of
the real game and, in view of the above hope, are taken by
progressively adding vehicles. The last column in
Table~\ref{test-cases} gives a brief description of the considered
game. The V and H prefixes refer to a vehicle put vertically or
horizontally, while the coordinates are those of the rows (counted
from top to bottom) and columns (counted from left to right).

Table~\ref{guarded-list-performance} reports on the data obtained on a
portable computer Lenovo x64 bits, running Windows 10 with 16 GB of
memory. The first column refers to the test case, the second and the
third columns give the time in milliseconds necessary for model
checking, the fourth column the time ratio and the last column the
hoped gain according to $2^n$ where $n$ is the number of vehicles in
the game. As can be seen from this table, guarded lists lead to a real
performance gain and even a greater performance than
expected\footnote{In the last case, we stopped the model-checker after 10 hours of run}.
This can be explained by the fact that the \Scan\ and \Anemone\ model
checker relies on non-optimized structures like sequential lists and
basically evaluates dynamically the transition system during the
model-checking phase. It is also interesting to observe that the
exponential behavior resulting from the interleaving of behaviors is
kept to a reasonable cost for the first five cases with guarded lists,
while it starts exploding from the fourth case without guarded lists.
The interested reader may redo the campaign of tests by using the
material available at \cite{Rush-Hour-Program}.

\section{Related work}
\label{relatedWork}

Although, to the best of our knowledge, it has not been exploited by 
coordination languages, the idea of forcing statements to be executed
without interruption is not new. In \cite{Dijkstra-GC} Dijkstra has
introduced guarded commands, which are statements of the form of
$G \rightarrow S$ that atomically executes statement $S$ provided the
condition $G$ is evaluated to true. They are mostly combined in
repetitive constructs of the form
\[ \begin{array}{ll}
  \textbf{do} & G_0 \rightarrow S_0 \\
  \Box & G_1 \rightarrow S_1 \\
       & \cdots \\
  \Box & G_n \rightarrow S_n \\
  \textbf{od}
\end{array} \]
which repetitively selects one of the executable guarded commands
until none of them are executable. A non-deterministic choice is
operated in the selection of the guarded commands in case several of
them can be executed. Later Abrial has used guarded commands in
the Event-B method \cite{Event-B-book}. Such a construct is also at
the core of the guarded Horn clause framework proposed by Ueda in
\cite{Ueda-GHC} to introduce parallelism in logic programming. There
Horn clauses are rewritten in the following form
\[ H \leftarrow G_1, \cdots, G_m | B_1, \cdots, B_n
\]
with $H$, $G_1$, \ldots, $G_m$, $B_1$, \ldots, $B_n$ being atoms.  The
classical SLD-resolution used to reduce an atom is modified as
follows. Assume $A$ is the atom to be reduced. All the clauses whose
head $H$ is unifiable with $A$ have their guard $G_1, \cdots, G_m$
evaluated. The first one which succeeds determines the clause that is
used, the other being simply discarded. To avoid mismatching
instantiations of variables, the evaluation of any $G_i$ is suspended
if it can only succeed by binding variables.  Finally, several pieces
of work have tried to incorporate transactions and atomic constructs
in ``classical'' process algebras, like CCS. For instance,
A2CCS~\cite{Gorrieri-A2CCS} proposes to refine complex actions into
sequences of elementary ones by modelling atomic behaviors at two
levels, with so-called high-level actions being decomposed into atomic
sequences of low-level actions. To enforce isolation, atomic sequences
are required to go into a special invisible state during all their
execution. In fact, sequences of elementary actions are executed
sequentially, without interleaving with other actions, as though in a
critical section.  RCCS~\cite{Danos-RCCS} is another process algebra
incorporating distributed backtracking to handle transactions inside
CCS. The main idea is that, in RCCS, each process has access to a log
of its synchronization history and may always wind back to a previous
state. A similar idea of log is used in
AtCCS~\cite{Acciai-AtCCS}. There, during the evaluation of an atomic
block, actions are recorded in a private log and have no effects
outside the scope of the transaction until it is committed. An
explicit termination action ``end'' is used to signal that a
transaction is finished and should be committed. States are used in
addition to model the evaluation of expressions and can be viewed as
tuples put or retrieved from shared spaces in coordination
languages. When a transaction has reached commitment and if the local
state meets the global one, then all actions present in the log are
performed at the same time and the transaction is closed. Otherwise
the transaction is aborted.

Our guarded list construct share similarities with these pieces of
work. A major difference is however that we restrict the guard to a
single primitive to be evaluated. This eases the implementation since,
once the primitive has been successfully evaluated, the remaining
primitives can be executed in a row without using distributed
backtracking as in RCCS, private spaces as in AtCCS for speculative
computations and checks for compatibility between local and global
environments. Intricate suspensions inherent in guarded Horn clauses
are also avoided. Nevertheless, under this restriction, the
combination with the non deterministic choice operator $+$ allows to
achieve computations similar to the repetitive statements of guarded
commands. With respect to these pieces of work, our contribution is
also to focus on model checking and to propose a refinement strategy
that allows to transform programs by introducing the guarded list
construct. An expressiveness study is also proposed in this paper and
not in these pieces of work.

Limiting the state explosion problem in model checking by limiting
interleaving is similar in spirit with the partial-order reduction
introduced in \cite{Godefroid-91,McMillan-92,Peled-93,Valmari-96}.
Realizing that $n$ independent parallel transitions result in $n!$
different orderings and $2^n$ different states, the idea is to select
a representative composed of $n+1$ states. Indeed, as the transitions
are independent, properties need only to be verified on a possible
ordering. This technique has been employed in many research efforts
for model checking asynchronous systems. However, these efforts aim at
designing more efficient algorithms on optimized automata. The
approach taken here is different. We do not change our algorithm for
model checking, but rather introduce a new construct as well as
considerations on refinements to transform programs into more
efficient programs.

\section{Conclusion}
\label{conclusion}

In the aim of improving the performance of the model checking tool
introduced in the workbenches \Scan~\cite{Scan} and
\Anemone\cite{Anemone}, thĩs article has introduced a new construct,
named guarded list. It has been proved to yield an increase of
expressiveness to Linda-like languages, while indeed bringing an
increase of efficiency during the model checking phase. In order to
pave the way to transform programs by safely introducing the guarded
list construct, we have also proposed a notion of refinement and have
characterized situations in which one can safely replace a sequence of
primitives by a guarded list of primitives.

Our work opens several paths for future research. As regards the
expressiveness study, we have used the approach proposed in
\cite{BJ98} for a few sublanguages. This naturally leads to deepen the
study to include all the sublanguages and to compare them with the
$L_{MR}$ and $L_{CS}$ families of languages studied in
\cite{BJ98}. Moreover this approach is only one of the possible
approaches to compare languages. It would be for instance interesting
to verify whether the absolute approach promoted by Zavattaro et al in
\cite{Zavattaro-expressiveness} would change the expressiveness
hierarchy of languages. Moreover, expressiveness studies based on
bisimulations and fully abstract semantics such as reported in
\cite{Peters} are also worth exploring. As regards model-checking, the
algorithm embodied in the \Scan\ and \Anemone\ workbenches is quite elementary and
calls for improvements. In that line of research, it would be
interesting to study how state collapsing and pruning techniques used
for checking large distributed systems may improve the performance of
the model checker. Finally, future work will aim at developing further the theory of refinement
and in investigating correctness preserving transformation
techniques.

\section{Acknowledgment}

The authors thank the University of Namur for its support. They also
thank the Walloon Region for partial support through the Ariac project
(convention 210235) and the CyberExcellence project (convention
2110186). Moreover they are grateful to the anonymous reviewers for
their comments on earlier versions of this work.


\bibliographystyle{eptcs}
\bibliography{generic}

\end{document}